\documentclass[12pt]{article}

\usepackage{authblk}
\usepackage{amssymb}
\usepackage{amsmath}
\usepackage{mathrsfs}
\usepackage{amsfonts}
\usepackage{braket}
\usepackage{color}
\usepackage{graphicx}
\usepackage{tikz}
\usetikzlibrary{shapes,snakes}
\usepackage{pgfplots}
\pgfplotsset{compat=1.9} 
\newtheorem{thm}{Theorem}

\newtheorem{lem}[thm]{Lemma}

\newtheorem{exmp}[thm]{Example}
\newtheorem{defn}[thm]{Definition}
\newcommand{\norm}[1]{\left\Vert#1\right\Vert}

\def\XXint#1#2#3{{\setbox0=\hbox{$#1{#2#3}{\int}$ }
\vcenter{\hbox{$#2#3$ }}\kern-.6125\wd0}}

\newcounter{lastnote}

\title{Kinematics and Dynamics of Quantum Walks in terms of Systems of Imprimitivity}
\author{Radhakrishnan Balu}
\affil{Army Research Laboratory Adelphi, MD, 21005-5069, USA \\
       radhakrishnan.balu.civ@mail.mil \\
       Computer Science and Electrical Engineering, \\
       University of Maryland Baltimore County, \\
       1000 Hilltop Circle, Baltimore, MD 21250
       radbalu1@umbc.edu}            
\date{Date: \today}
\begin{document}
\maketitle

\begin{abstract}
We build systems of imprimitivity (SI) in the context of quantum walks and provide geometric constructions for their configuration space. We consider three systems, an evolution of unitaries from the group $SO_3$ on a low dimensional de Sitter space where the walk happens on the dual of $SO_3$, standard quantum walk whose SI live on the orbits of stabilizer subgroups (little groups) of semidirect products describing the symmetries of 1+1 spacetime, and automorphisms (walks are specific automorphisms) on distant-transitive graphs as application of the constructions.
\end{abstract}

\section{Introduction}
\label{intro}
The concept of localization, where the position operator is properly defined in a manifold, and covariance in relativistic sense of systems can be completely characterized by systems of imprimitivity. These are representations of a group induced by representations of subgroups, more specifically stabilizer subgroups at a point in the orbit of the subgroup. Systems of imprimitivity are a more fundamental characterization of dynamical systems, when the configuration space of a quantum systems is described by a group, from which infinitesimal forms in terms of differential equations ($Shr\ddot{o}dinger$, Heisenberg, and Dirac etc), and the canonical commutation relations can be derived. For example, using SI arguments it can be shown \cite {Wigner1949} that massless elementary particles with spin less than or equal to 1 can't have well defined position operators, that is photons are not localizable. As discrete quantum walks can lead to Dirac evolution in the continuum, with proper choice of coin parameters and specific initial conditions \cite {Chandra2010} and \cite {Rad2018}, in this work we construct systems of imprimitivity for such quantum walks and describe the governing differential equations.

Let us first define the notion of SI and an important theorem by Mackey that characterizes such systems in terms of induced representations.

\begin {defn} \cite {Varadarajan1985} A G-space of a Borel group G is a Borel space X with a Borel automorphism $\forall{g\in{G}},t_g:x\rightarrow{g.x},x\in{X}$ such that
\begin {align}
&t_e \text{ is an identity} \\
&t_{g_1,g_2} =t_{g_1}t_{g_2}
\end {align}
The group G acts on X transitively if $\forall{x,y\in{X}},\exists{g}\in{G}\ni{x=g.y}.$
\end {defn}
\begin {defn} \cite {Varadarajan1985} A system of imprimitivity for a group G acting on a Hilbert space $\mathscr{H}$ is a pair (U, P) where $P: E\rightarrow{P_E}$ is a projection valued measure defined on the Borel space X with projections defined on the Hilbert space and U is a representation of G satisfying
\begin {equation}
U_gP_EU^{-1}_g = P_{g.E}
\end {equation}  
\end {defn}
When the action of G on X is transitive it is called transitive system of imprimitivity. In a more generalized setting when P is instead a projection operator valued measure (POVM),a resolution of identity, it is called a covariant system \cite {JPG2000}.
\begin {thm} (Mackey{'}s Imprimitivty Theorem \cite {JPG2000}) Let {U, P} be a transitive system of imprimitivity, based on the homogeneous (transitive) space X of the locally compact group G. Then, there exist a closed subgroup H of G, the Hilbert space $\mathscr{H}$, and a continuous unitary representation V of H on $\mathscr{H}$, such that the given system is unitarily equivalent to the canonical system of imprimitivity $(\tilde{U}, \tilde{P})$, arising from representation $\tilde{U}$ of G induced from V.
\end {thm}
\begin {thm} \cite{JPG2000} if (U,P) is a transitive system of covariance, then U is a sub representation of an induced representation.
\end {thm}

The next step is to consider semidirect product of groups, that naturally describe the dynamics system of quantum walks, and use the representation of the subgroup to induce a representation in such a way that it is an SI.

\begin {defn}
Let A and H be two groups and for each $h\in{H}$ let $t_{h}:a\rightarrow{h[a]}$ be an automorphism (defined below) of the group A. Further, we assume that $h\rightarrow{t_h}$ is a homomorphism of H into the group of automorphisms of A so that
\begin {align} \label{semidirectEq}
h[a] &= hah^{-1}, \forall{a\in{A}}. \\
h &= e_H, \text{  the identity element of H}. \\
t_{{h_1}{h_2}} &= t_{h_1}t_{h_2}.
\end {align}
Now, $G=H\rtimes{A}$ is a group with the multiplication rule of $(h_1,a_1)(h_2,a_2) = (h_1{h_2},a_1{t_{h_1}}[a_2])$. The identity element is $(e_H,e_A)$ and the inverse is given by $(h,a)^{-1} = (h^{-1},h^{-1}[a^{-1}]$. When H is the homogeneous Lorentz group and A is $R^4$ we get the Poincare group via this construction.
\end {defn}

\section {Coins Space Ensembles}
The group $SO_3$ is a compact Lie group for which we can apply the above construction to build systems of imprimitivity to describe a relativistic quantum walk on de Sitter space. In this walk evolution the coins are drawn from the compact group $SO_3$ and the statistics of the coins (affected by the walker DOF due to spin-orbit coupling) are considered. Alternately, an evolution $\{U^n_g, g\in{SO_3}\}$ can thought of as a walk on a non commutative space whose three dimensional axes would correspond to the generators of the Lie algebra $so_3$ \cite {Biane2008}.

In Ref \cite{Hannabuss1969} the author uses an elegant technique that takes the difference between Lie algebra operators of two representations of the group $SO_5$ that are manifestly covariant under the actions of the group.  This leads to a simple derivation of Dirac equation on curved spacetime, on two spaces with each symmetries preserving the metrics (-1,-1,-1,+1,-1) and (-1,-1,-1,+1,+1), described by the equation $x_1^2 + x_2^2 + x_3^2 +x_4^2+x_5^2= L$ in space $R^5$ that has three space like and two time-like coordinates. We adopt that method for a similar construction on $R^3$ with a space coordinate and two time-like coordinates.

The 3 generators, 2 boosts and a rotation, of the Lie algebra $so_3$ are $J_{ab},a,b\in\{1,2,3\}$ with the only possible Casimir operator $J^2$ (corresponding to spin angular momentum) for this algebra. The invariant operators (Casimir) are the center of the Lie algebra $\textit{so}_3$. Let us induce an unitary representation of it by using a finite dimensional representation $\Sigma$ of $SO_3$, could be any one of the countably infinitely many unitary irreducible representations afforded by Plancherel decomposition \cite {KP1990}, restricted to its subgroup $SO_2$. Now, we have a representation of $SO_3 = U^1\times{\Sigma}$ that is a product of the one induced by the trivial representation of $SO_2$ (which is just the regular representation) and $\Sigma$. 

The infinitesimal (differential), $M_{ab}$ in representation $U^1$ and $\sigma_{ab}$ that of $\Sigma$, form of the generator and Casimir operators for the Lie algebra of the group $U^1\times{\Sigma}$ can be written as the angular momentum operators that are more fundamental than linear momentum operators in de Sitter space \cite {Dirac1935}
\begin {align}
M_{ab} &= i\{x_a(\frac{\partial}{\partial{x_b}}) - x_b(\frac{\partial}{\partial{x_a}})\}. \\
 J_{ab} & = M_{ab}\times{\mathbb{I}} + \mathbb{I}\times\sigma_{ab}\\
 J^2_{ab} & = M_{ab}M_{ab} + 2M_{ab}\sigma_{ab} + \sigma_{ab}\sigma_{ab}.
\end {align}
The invariant operator $J^2$ can be used to derive Klein-Gordon equation, to establish Dirac equation using first-order operators we need the following procedure that uses induced representations for building SI to guarantee covariance.
The expression $M_{ab}$ commutes with $J_{ab}$ (they are on different Hilbert spaces) and so the reminder $2M_{ab}\sigma_{ab} + \sigma_{ab}\sigma_{ab}$ of $J_{ab}^2$ also commutes with every operator of the algebra giving rise to the following Dirac equation with first order operators as a result of Schu{r'}s lemma:
\begin {equation}
\left(M_{ab}\sigma_{ab} + \frac{1}{2}\sigma_{ab}\sigma_{ab}\right)\Psi = \lambda\Psi.
\end {equation}
In the neighborhood of $(0, 0, R^2)$ the above equation may be written as
\begin {equation} \label {dsDirac}
\{p_2 + \sigma_y{p_1} - m\sigma_y\sigma_z\} \Psi= 0.
\end {equation}
Later, we will provide a geometric interpretation for this construction in terms of fiber bundles.

\section {Quantum Walks}
Next, let us apply the constructions to quantum walk evolutions. Quantum walks are unitary evolutions that involve a coin Hilbert space and a walker Hilbert space where the dynamics happens \cite{Kempe2003}, \cite {Rad2017a}, \cite {Rad2017b}. Let $\mathscr{C}^2$ (complex space) and $\mathbb{Z}$ (set of integers) correspond to Hilbert spaces of the coin and walker respectively. The dynamics of the quantum walk is described by the unitary operator U composed of a rotation on the Bloch sphere and a translation on the integer line
\begin {align} 
L^{\pm}(x) &= x\pm{1}, x\in{\mathbb{Z}}. \\
\mathscr{C}^2 &= \Pi_0\oplus{\Pi_1},\Pi_i = \ket{i}\bra{i}. \\
S(x) &= \Pi_0\otimes{L^+}+\Pi_1\otimes{L^-}. \\
U(x) &= S(x)T.  \label {QWalk}
\end {align} 
That is, the walker takes a step to the right on $\mathbb{Z}$ if the outcome of projecting the Hadamard coin 
\begin {align*}
T &= \frac{1}{\sqrt{2}}\begin {bmatrix} 1 & 1 \\
                                 1 &  -1 \\
                  \end {bmatrix}, \\                 
\end {align*}
on $\mathscr{C}^2$ is 0 and moves a step left otherwise. A variation of the above unitary evolution split-step quantum walk \cite {Rad2018a} can be defined as follows:
\begin {align}
S(x)^+ &= \Pi_0\otimes{L^+}+\Pi_1\otimes\mathbb{I}. \\
S(x)^- &= \Pi_0\otimes{\mathbb{I}}+\Pi_1\otimes{\L^-}. \\
W(x) &= S^-(x)T(\theta_2)S^+(x)T(\theta_1).  \label {SS-QWalk}
\end {align}
The split-step quantum walk has an effective Hamiltonian $H_{eff}$ with several symmetries: 
\begin {align*}
\Gamma_\theta &= e^{i\pi{A\theta}.\frac{\sigma}{2}}. \\
\Gamma_\theta^{-1}H_{eff}\Gamma_\theta &= -H_{eff}. &\text {  Chiral symmetry.} \\
P_h &= K &\text {  Complex Conjugation}. \\
P_{h}H_{eff}P_h &= -H_{eff}. &\text {  Particle-Hole symmetry.} \\
T &= \Gamma_\theta{P_h}. \\
TH_{eff}T^{-1} &= H_{eff}. &\text {Time-reversal symmetry.}
\end {align*}
In this work we describe the kinematics of quantum walks in terms induced representations of groups that act on the configuration space of the walker and derive systems of imprimitivity.

\begin {thm} Split-step quantum walk described by the equation \eqref {SS-QWalk} is a transitive system of imprimitivity. 
\end {thm}
\begin {proof} The Hamiltonian of the walk is translational invariant simultaneously with respect to the time and X axis. The configuration space of the 1+1 spacetime quantum walker can be seen to satisfy the criteria for periodic lattices in Minkowski space \cite {Boozer2010}. Boozer has listed viable spacetime lattices that are Lorentz invariant specified by parameters such as the ratio between the lattice constants of space and time dimensions and one feasible value is $\sqrt{3}$. The rest of the parameters are concerned with choosing their values for the frames of references for Lorentz variance. Let us observe that the 1+1 spacetime configuration of the walker is $\mathbb{Z}^2$  whose Pontryagin dual (Fourier space) is 2-torus that has the symmetry described by the semidirect product of two groups, the locally compact discrete Lorentz group $O(1,1)$ and the abelian 2-torus $\mathbb{T}^2$.
The same spacetime posses symmetry described by the  semidirect product of two groups, the locally compact discrete Lorentz group $SO_3$, whose universal cover is $SU_2$, and the abelian 2-torus $\mathbb{T}^2$ as the generic Lorentz transformation can be written as a sum of a transformations on coordinate and internal degrees of freedom \cite {Kim1991}. The internal DOFs may be the spinors in the case of massive particles and for massless photons they are rotation around momentum and gauge transformations. Discrete Lorentz group $O(1,1)$ \cite {Schwarz1976} with boost, with a generic element $\Lambda = \begin{bmatrix} cosh\phi & sinh\phi \\ sinh\phi & cosh\phi \end{bmatrix}$, is the only transformation apart from translations in this space is an automorphism group for the 2-torus and so the semidirect product $O(1,1)\rtimes\mathbb{T}^2$ is well defined.

Let $\mathscr{H}$ be a separable Hilbert space and P is a projection valued measure based on $\mathbb{T}^2$ acting in $\mathscr{H}$.
\begin {align*}
\langle{U}_a{f}, f_1\rangle) &= \int_{\mathbb{T}^2} x(a)d\nu_{ff_1}(x), x\in{\mathbb{T}^2}, \text{  x(a) is the value of character at a}. \\
\nu_{ff_1}(E) &= \langle{P}_E, f_1\rangle. \\
U_a &=  \int_{\mathbb{T}^2} x(a)dP(x).
\end {align*}
group representation S of equation \eqref {QWalk} acts on $\mathbb{Z}^2$ give rise to a system of imprimitivity.
Suppose V is a unitary representation of O(1,1). Then, by equation \ref {semidirectEq} we have,
\begin {align*} 
hah^{-1} &= h[a]. \\
V_h{U_a}V_h^{-1} &= U_{h[a]}; h\in{O(1,1)}, a\in\mathbb{T}^2. \text { replacing the above by its unitary representation}
\end {align*}
Using the spectral resolution of $U_a$ we can write the left hand side of the above equation as a representation of $a\in\mathbb{T}$ and so the representation of E is $V_h{P_E}V_h^{-1}$. Similarly, the mapping $a\rightarrow{h[a]}$ gives the representation $E \rightarrow{P_{h[E]}}$ and we have the following system of imprimitivity relationship.
\begin {align*} 
V_h{P_E}V_h^{-1} &= P_{h[E]}.
\end {align*}
$\blacksquare$
\end {proof}

\section {Little groups (stabilizer subgroups)}
Some of the different systems of imprimitivity that live on the orbits of the stabilizer subgroups are described below. It is good to keep in mind the picture that SI is an irreducible unitary representation of Poincare group $\mathscr{P}^+$ induced from the representation of a subgroup such as $SO_3$ as $(U_m(g)\psi)(k) = e^{ik.a}\psi(R^{-1}_pk)$ where g belongs to the Poincare group and R is a member of the rotation group and the expression is in momentum space.

The stabilizer subgroups of the Poincare group $\mathscr{P}^+$ can be described as follows: \cite {Kim1991}.
Time-like quantum walker: There is a reference frame in which the 2-component momentum is proportional to (0,m) and the stabilizer subgroup is $SO_3$ that describes the spin. The walker, a massive particle, is rest in this frame. Let us denote the eigen vector of the Casimir operator $P_\mu P^\mu$ as  $\ket{0\lambda}$. Then, we can describe the invariant space under the Poincare group $\mathscr{P}^+$ as $\Lambda_p\ket{0\lambda} = \ket{p\lambda}$ by applying the Lorentz boost. Any Lorentz operator operating on this space can be shown to be by a rotation (spin in our case). In other words, quantum walk is covariant if a proper basis is chosen to describe the spacetime grid such as the Boozer lattices \cite {Boozer2010}on which the walker evolves.

Space-like walker: The Lorentz frame in which the walker is at rest has momentum proportional to (Q,0) and the little group is
again $SO_3$ and this time the rotations will change the helicity. In this imaginary mass case the little group is rotations around the space axis and the analysis above carries through.

Light-like walker: There is no frame in which the walker is at rest but the frame where the momentum is proportional to $(\omega,\omega)$ has the stabilizer subgroup with elements of the form $J_1$, that is a rotation around the first component of momentum and the boost $\Lambda_p$ in the spatial direction \cite {Kim1991}. These two operators commute and the induced representation can be constructed as above. In the case of light-like particle the induced representation for the 3+1 spacetime Poincare group was constructed in phase space by Kim et al, \cite {Kim1991} and for the case of 1+1 dimension more involved construction involving fiber bundles are required \cite {Ali1993}. The group generates the coherent states (frames in mathematical context) is the same that makes the $(\omega,\omega)$ momentum vector invariant. The tools of frames on Hilbert spaces theory \cite {Schr1926} \cite {Glauber1963} can be applied to generate Lorentz invariant inertial frames of quantum walkers.

\
\section {Quantum Graphs} 
\begin{defn} A finite dimensional quantum probability (QP) space is a tuple $(\mathscr{H},\mathbb{A}, \mathbb{P})$ where  $\mathscr{H}$ is a finite dimensional Hilbert space, $\mathbb{A}$  is a *-algebra of operators, and $\mathbb{P}$ is a trace class operator, specifically a density matrix, denoting the quantum state.\end{defn}
 Given a graph G = (V, E) with finite number of vertices V and set of edges E let us consider the Hilbert space $\mathscr{H}=l^2(V)$, the square summable functions on V with the inner product $\langle{f}, g\rangle = \sum_{x\in{G}}{\overline{f(x)}g(x)}, g,f\in{l^2(G)}$.
An edge $(x,y)\in{E}, x,y\in{V}$ is denoted by $x\sim{y}$ and the distance function $\partial(x, y), x,y \in{V}$ is defined as the shortest path connecting the two vertices.\\
The collection of functions below form an orthonormal basis of $l^2(G)$.
\begin {align*}
\delta_x(y) &= 1 \text{   if x = y}.\\
                 &= 0  \text{   o.t.}
\end {align*}
Let $C_0(V)$ be the dense subspace of $l^2(G)$ spanned by $\{\delta_x\}$ and the adjacency algebra $L(C_0(V))$ be the *-algebra generated by linear operators on $C_0(V)$. The operator T is said to be locally finite if the following conditions are satisfied:
\begin {equation}
|\{x\in{V}: T_{xy} = 0\}| < \infty, \forall{y\in{V}}.\\
|\{y\in{V}: T_{xy} = 0\}| < \infty, \forall{x\in{V}}.
\end {equation}
When G is finite the matrix elements of $T\in{L(C_0(V))}$ are determined as $T_{xy} = \langle{\delta_x}, T\delta_y\rangle$ and its adjoint as $T^{\dag}_{xy} = \langle{T\delta_x}, \delta_y\rangle$.
\begin {defn}
Given a graph $\mathscr{G} = (G, E)$ and an adajacency algebra $\mathscr{A}(\mathscr{G})$ a vacuum state $\delta_o$ at a fixed origin of the graph $o \in {V}$ is defined as $\langle{\delta_o}, a\delta_o\rangle, a\in {\mathscr{A}}.$
\end {defn}
It is easy to verify the fact $(A^m)_{xy} = \langle\delta_x, A^m\delta_y\rangle$ which is the number of m-steps walks connecting x and y vertices.
\begin {defn}
The adjacency matrix of a graph (G, E) is defined as 
\begin {align*}
A_{xy} &= 1  \text{    if (x, y) }\in {E}. \\
            & = 0 \text{   o.t.}
\end {align*}
\end {defn}
It is easy to see that A is symmetric, taking values in \{0, 1\}, with vanishing diagonals. 
\begin {defn}
A stratification of a graph $\mathscr{G} = (V, E)$ with a fixed origin $o \in {V}$ is defined a disjoint union of strata
\begin {equation}
V = \bigcup_{n = 0}^\infty {V_n}, V_n = \{x: \partial(o, x) = n\} \label{eq:strat}
\end {equation} 
\end {defn}
It is easy to verify that 
\begin {equation}
\forall{x,y\in{V}}, x\sim{y} \text{  and  } x\in{V_n} \Rightarrow {y}\in{V_{n-1}}\cup{V_n}\cup{V_{n+1}}
\end {equation}

We define the following matrices w.r.t the stratification equation (\ref{eq:strat}) and the quantum decomposition of A
\begin {align}
\forall{y\in{V_n}}, (A_{xy})^{\epsilon} &= A_{xy} \text{   } x\in{V_{n+\epsilon}} \text{   , } \epsilon\in\{o,+,-\}.\\
                             &= 0 \text{  o.t.} \\
                          A &= A^{o} + A^{+} + A^{-}. \\
                          (A^+)^* &= A^-. \\
                          (A^-)^* &= A^+. \\
                          \langle{A^{+}f},g\rangle &= \langle{f},A^{-}g\rangle,\forall{f,g\in{C_0{V}}}.\\
                          (A^o)^* &= A^o
\end {align}
Given a stratification of a graph with origin $o\in{V}$ the degree k(x) of a vertex $x\in{V}$ can be decomposed as follow
\begin {align}
\omega_{\epsilon}(x) &= |\{y\in{V}; x\sim{y},\partial(o, y) = \partial(o, x) + \epsilon\}|. \\
k(x)                            &= \omega_{o}(x) + \omega_{+}(x) + \omega_{-}(x).
\end {align}
The *-algebra generated by $A^{\epsilon}$ is a non commutative one. The vectors $\Phi_n=\sum_{x\in{V_n}}\delta_x$ span a space denoted by $\Gamma(\mathscr{G})$.
\begin {defn}
A graph $\mathscr{G} = (V, E)$ is called distance-regular if for any choice of $x,y \in {V}$ with $\partial(x, y) = k$ the number of vertices $z \in {V}$ such that $\partial(x, z) = i$ and $\partial(y, z) = j$ is independent of the choice of x and y. Then, the intersection numbers i,j,and k are defined as 
\begin {equation}
p^k_{ij} = |\{z\in{V}: \partial(x, z)=i, \partial(y, z) = j, \partial(x, y) = k\}|.
\end {equation}
\end {defn}
\begin {defn}
A k-th distance adjacency matrix of a graph $\mathscr{G} = (V, E)$ is defined as
\begin {align}
(A^k)_{xy}     &=   1 \text{  if  } \partial(x, y) = k. \\
                      & = 0 \text{   o.t.} \\
 \sum_k{A^K} & =   J \text{,  where  } (J)_{xy} = 1.
\end {align}
\end {defn}
\begin {lem}
Let $\mathscr{G} = (V, E)$ be a distance-regular graph with intersection numbers $p^k_{ij}$. Then the following holds:
\begin {equation}
A_{i}A_{j} = \sum_{|i - j|}^{i + j} {p^k_{ij}A_k}.
\end {equation}
\end {lem}

Association schemes are naturally related to graphs and the resulting quantum probability space. Specifically, the Bose-Mesner algebra can be identified with distance regular graphs \cite {Obata2007} that leads to our next result.

\begin {thm} Let $(\chi = V, A_i, 0\leq{i}\leq{d}, R_i)$ be an association scheme of Bose-Mesner type and the corresponding adjacency matrix is defined on a distance-regular graph $\mathscr{G} = (V, E)$ whose vertices V with cardinality $\|V\| = d$ such that 
\begin {align*}
(A_i)_{xy} &= 1. \text{  if  } (x,y)\in {R_i}. \\
                 &= 0 \text{  } o.t.
\end {align*}
 Let $o \in {V}$ be the origin of the graph and $\{\Phi_i, i = 1, \dots, n\}$ as defined above. Then, the unitary regular representation $U_g(f(x)) = f(g^{-1}x)$ of the group G on the Hilbert space $\mathscr{H} = L^2(V)$ and the projection valued measure $P = \{\ket{\Phi_i}\bra{\Phi_i}, i=1, \dots, n\}$ form a transitive system of imprimitivity. 
 \end {thm}
\begin {proof}
Let us first observe that G being a finite group it is compact and its left and right regular representations coincide. Besides, the regular representation is induced by the trivial representation of its subgroup {e} and so it is transitive. G being a distance-regular graph ($\Gamma(G), \{\Phi\}_n, A^+,A^-,A^0$) is an interacting Fock space with the Jacobi sequence 
\begin {align*}
\omega_n &= \frac{|V_n|}{|V_{n-1}|}\omega_{-}(y)^2,\text{   }y\in{V_n}. \\
\omega_{\epsilon}(x) &= |\{y\in{V};y~x,\delta(o,y)=\delta(x)+\epsilon\}|.\\
\alpha_n &= \omega_o(y), \text{    } y\in{V_{n-1}}, n=1,2,...
\end {align*}

\begin {align}
U_gP_{\Phi_n}U^{\dag}_g (f(y)) &= U_gP_{\Phi_n}(f(gy)), f \in l^2(V), g,y \in {G}.\\
& = U_g \left(\sum_{x\in{V_n}}\delta_x\right)\left(\sum_{x\in{V_n}}\delta_x\right)^{\dag} (f(gy)). \\
& = U_g \left(\sum_{xg^{-1}\in{V_n}}\delta_x\right)\left(\sum_{xg^{-1}\in{V_n}}\delta_x\right)^{\dag} (f(y)). \\
& = \left(\sum_{gxg^{-1}\in{V_n}}\delta_x\right)\left(\sum_{gxg^{-1}\in{V_n}}\delta_x\right)^{\dag} (f(y)). \\
&= P_{g.{\Phi_n}}.
\end {align}
$\blacksquare$
\end {proof}

\section {Geometric interpretation and the Dirac equation}

The states of a freely evolving relativistic quantum particles are described by unitary irreducible representations of Poincare group that has a geometric interpretation in terms of fiber bundles. To express quantum walks in a similar fashion let us recall the related definitions.

\begin {defn} A fiber bundle is a triple $(E,\pi,M)$ where E is a space, $\pi:E\rightarrow{M}$ is a projection from E to a manifold M. A principle fiber bundle is isomorphic to $(E,\rho,E/G)$ where E is a right G-space of a Lie group, G acts on E freely, and the quotient group $E/G$ is the orbit space of the G-action on E. A bundle can be thought of as a product space with a twist which can encode  information such as topological invariants or Casimir operators. 
\end {defn}

\begin {exmp} Let G be the compact Lie group SO(3) and H is the closed subgroup SO(2), then the group acts freely on the H-orbits space, which is four-sphere, and $(G,\pi,G/H)$ is a principle H-bundle. The bundle can be viewed as a twisted product of H and G/H.
\end {exmp}

\begin {defn} A cross-section of a bundle $(E,\pi,M)$ is a map $s: M\rightarrow{E}$ such that the image of each point $x\in{M}$ lies in the fiber $\pi^{-1}(x)$ over x as $\pi\circ{s} = id_M$. When these functions are square integrable with respect to an appropriate measure they form a Hilbert space to describe the states of quantum systems. In general, fiber bundles may not have cross-sections but for vector bundles they can be constructed. An important theorem states that principle bundles have smooth cross sections only when they are trivial, that is the twist is a regular product \cite {Isham1989}. 
\end {defn}
\begin {defn} Let X and Y be G-spaces, then a G-product $\times_G$ is the equivalence class $(x,y)\equiv({x'},{y'})$ if $\exists{g\in{G}},\ni$ {x'}=gx and {y'} = gy and is denoted by $X\times_G{Y},[x,y]$. They are the G-orbits contained in the product space.
\end {defn}

\begin {defn} Let $\eta = (P,\pi,M)$ be a principle G-bundle and F be a left G-space. Define the F-orbits as $P_F = P\times_F{G}$ where the left action is $g(p,v) = (pg, g^{-1}v)$ and a projection $\pi_{F}:P_F\rightarrow{M}$ by the map $\pi_F([p,v]):=\pi(p)$. Then $\eta[F] = (P_F,\pi_F,M)$ is the fiber bundle associated with the principle bundle $\eta$ via the action of the group G on F. It is easy to prove that $\forall{x\in{M}}$ the space $\pi_F^{-1}(x)$, the fibers, is homeomorphic to F.  The specific left action on the F-orbits are required to make sure that the associated bundle projection $\pi_F$ is well defined. This can be seen with if $[p_1, v_1] = [p_2, v_2]$ then $\exists{g}\in{G}$ such that $(p_2, {v_2}) = (p_1{g}, g^{-1}v_1)$. This implies that 
\begin {equation}
\pi_F(p_2, v_2) = \pi(p_2) = \pi(p_1{g}) = \pi(p_1) = \pi_F(p_1, v_1).
\end {equation}
\end {defn}
 This way a family of new fiber bundles can be constructed out of a principle G-bundle using G-spaces that act as fibers.

\begin {defn} A vector bundle is a special case of an associated bundle in which the fiber is a vector space. Furthermore, $\forall{x\in{M}}$there must exist some neighborhood $U\subset{M}$ of x and a local trivialization $h: U\times\mathscr{R}^n\rightarrow\pi^{-1}(U)$ such that, $\forall{Y}\in{U}, h:\{y\}\times\mathscr{R}^n\rightarrow\pi^{-1}(y)$ is a linear map.
\end {defn}

The discrete time quantum walk described above leads to entanglement between the internal spin and walker DOFs and when the initial state of the walker is at a highly localized (Compton wavelength) single site of the 1-D lattice with positive-energy it leads to relativistic propagation \cite {Frederick2006}. The 1+1 spacetime discrete Lorentz group $\hat{O}(1,1)$-orbits  of the momentum space $\mathbb{T}^2=S^1\times{S^1}$, where the systems of imprimitivity established above will live, described by the symmetry $\hat{O}(1,1)\rtimes\mathbb{T}^2$. That is, the rest frames of the walker are related by Lorentz boosts where the position oerator is defined up to Compton wavelength. The orbits have an invariant measure $\alpha^+_m$ whose existence is guaranteed as the groups and the stabilizer groups concerned are unimodular and in fact it is the Lorentz invariant measure $\frac{dp}{p_0}$ for the case of forward mass hyperboloid.  The orbits are defined as:
\begin {align}
\hat{X}^+_m &= \{p: p^2_0 - p^2_1 = m^2, p_0 > 0\}, \text{\color{blue} forward mass hyperboloid: positive-energy}.\\
\hat{X}^+_m &= \{p: p^2_0 - p^2_1 = m^2, p_0 < 0\}, \text{\color{blue} backward mass hyperboloid}. \\
\hat{X}_{00} &= \{0\}, \text{\color{blue} origin}.
\end {align}
Each of these orbits are discrete and invariant with respect to $\hat{O}(1,1)$ and let us consider the stabilizer subgroup of the first orbit at p=(m,0). Now, assuming that the spin of the particle is 1/2 let us define the corresponding fiber bundles (vector) for the positive mass hyperboloid that corresponds to the positive-energy states by building the total space as a product of the orbits and the group $SL(2, C)$ (more precisely the members of the Clifford algebra) and using the discrete form of Dirac equation in momentum space.
\begin {align}
\hat{B}^{+,\frac{1}{2}}_m &= \{(p,v) \text{   }p\in{\hat{X}^+_m, }\text{   }v\in\mathscr{C}^2,p_0\sigma_z + p_1\sigma_x = mv\}. \\
\hat{\pi} &: (p,v) \rightarrow {p}. \text{  Projection from the total space }\hat{B}^{+,\frac{1}{2}}_m \text{ to the base }\hat{X}^+_m.
\end {align}
In other words we have used a discretized Dirac operator to locally trivialize the fiber bundle and still it provides a representation of the group.As a consequence of the trivialization we have spin-orbit coupling in the system. The states of the particles are defined on the Hilbert space $\hat{\mathscr{H}}^{+,\frac{1}{2}}_m$, square integrable functions on Borel sections of the bundle $\hat{B}^{+,\frac{1}{2}}_m$ with respect to the invariant measure $\alpha^+_m$, whose norm induced by the inner product is given below:
\begin {equation}
 \norm{\phi}^2 = \int_{X^+_m}p_0^{-1}\langle\phi{p},\phi{p}\rangle.{d\alpha}^+(p).
\end {equation}

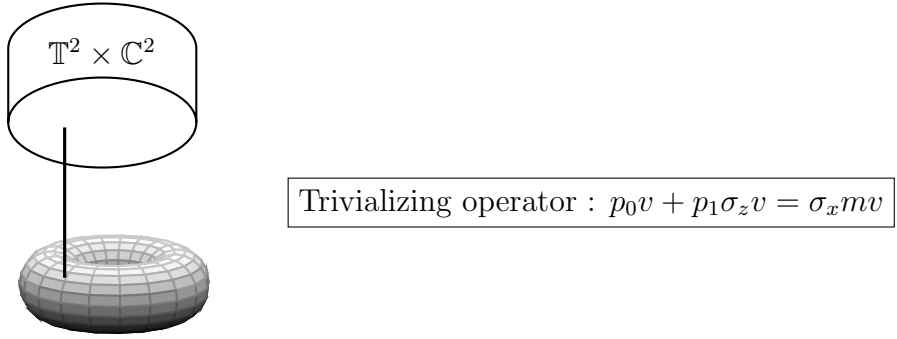
\begin{figure}
 \begin {center}
 \begin{tikzpicture}[scale=2]
     \node[cylinder,draw=black,thick,aspect=0.7,
        minimum height=1.7cm,minimum width=2.5cm,
        shape border rotate=-80,
        cylinder uses custom fill,
        ]  at (0.75,2.0)
   (A) {$\mathbb{T}^2\times\mathbb{C}^2$}; 
   \begin{axis}[hide axis, xtick=\empty, ytick=\empty,      
       width=3.25cm,height=2.5cm]
       \addplot3[surf,
       colormap/blackwhite, 
       samples=20,
       domain=0:2*pi,y domain=0:2*pi,
       z buffer=sort]
       ({(2+cos(deg(x)))*cos(deg(y+pi/2))}, 
        {(2+cos(deg(x)))*sin(deg(y+pi/2))}, 
        {sin(deg(x))});    
    \end{axis}  
     \node [draw]  at (4,1) {Trivializing operator : $p_0 v + p_1\sigma_z v = \sigma_xmv$};  
     \draw [draw=black,very thick] (0.5,1.5) -- (0.5,0.5);
\end{tikzpicture}
\end {center}
\caption {Fiber bundle formed with an hyperboloid, dense set of countably infinite points, of the 2-torus $\mathbb{T}^2$ as the base and $\mathbb{SL}(2,C)$ as fibers. }
\end {figure}

In the discrete configuration space a Dirac like difference equation of motion can be defined only in the momentum representation. Now, let us consider the continuous version of the above evolution with the 1+1 spacetime discrete Lorentz group $O(1,1)$-orbits  of the momentum space $\mathbb{R}^2$, where the systems of imprimitivity established above will live in the continuum of 1+1 spacetime, described by the symmetry $O(1,1)\rtimes\mathbb{R}^2$,with an invariant measure $\alpha^+_m$ can be defined as:
\begin {align}
X^+_m &= \{p: p^2_0 - p^2_1 = m^2, p_0 > 0\}, \text{\color{blue} forward mass hyperboloid}.\\
X^+_m &= \{p: p^2_0 - p^2_1 = m^2, p_0 < 0\}, \text{\color{blue} backward mass hyperboloid}. \\
X_{00} &= \{0\}, \text{\color{blue} origin}.
\end {align}

\begin {align}
B^{+,\frac{1}{2}}_m &= \{(p,v) \text{   }p\in{X^+_m, }\text{   }v\in\mathscr{C}^2,p_0\sigma_z + p_1\sigma_x = mv\}. \\
\pi &: (p,v) \rightarrow {p}. \text{  Projection from the total space }B^{+,\frac{1}{2}}_m \text{ to the base }X^+_m.
\end {align}

We have used the Dirac operator to locally trivialize the fiber bundle in this case. In the continuum an inverse transformation will provide the 1-D Dirac equation in position space as well. The above construction is similar to that of Poincare group \cite{Varadarajan1985} which was based on the original work of Wigner \cite {Wigner1939} that now leads to one dimensional Dirac equation after assuming the velocity of light, mass, and the Planck constant as unity. 
\begin{figure}
 \begin {center}
 \begin{tikzpicture}[scale=2]
  \node[cylinder,draw=black,thick,aspect=0.7,
        minimum height=1.7cm,minimum width=2.5cm,
        shape border rotate=-80,
        cylinder uses custom fill,
        ]  at (0,0)
   (A) {$\mathbb{R}^2\times\mathbb{C}^2$};  
    \node[rectangle,draw=black,thick,aspect=0.7,
        minimum height=1.4cm,minimum width=2.5cm,
        shape border rotate=60,
        cylinder uses custom fill,
        ]  at (0,-1.4)
   (B) {$\mathbb{R}^2$}; 
    \node [draw]  at (3,-1) {Trivializing operator : $p_0 v + p_1\sigma_z v = \sigma_xmv$};
   \draw [draw=black,very thick] (-0.2,-0.42) -- (-0.2,-1.4);
\end{tikzpicture}

\end {center}
\caption {Fiber bundle formed with an hyperboloid of $\mathbb{R}^2$ as the base and $\mathbb{SL}(2,C)$ as fibers. When the spinors that trivialize the bundle are replaced by wave functions and the momentum variables by their differential operators we get the Dirac equation in the familiar form $i\partial_t\Psi + i\sigma_x\partial_x\Psi = \sigma_z m\Psi$.}
\end {figure}
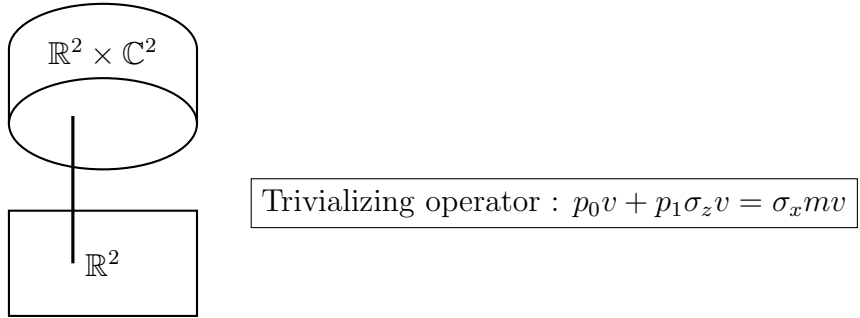
The limit of the 1+1 lattice is the euclidean space $R^2$ and so the torus based fiber bundle converges to the $R^2$ based bundle.
Another way to look at the limit is the dense set of points of the hyperbola of the discrete walk become continuous curve in the base space of the fiber bundle that is lifted to the bundle making the difference equation a differential one.
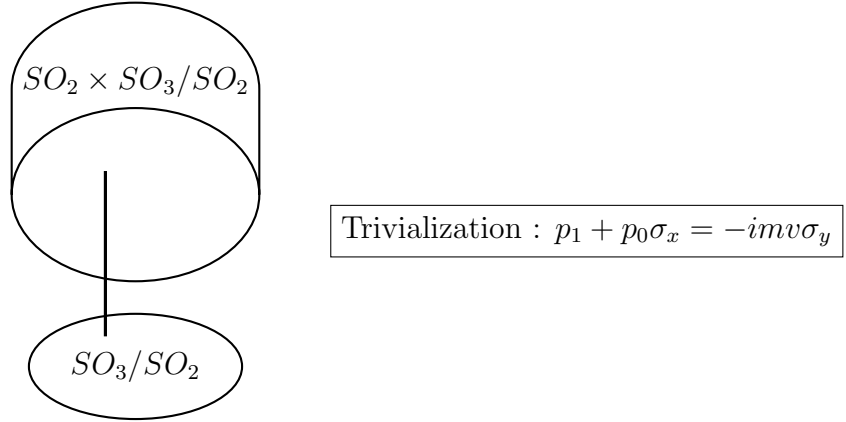
\begin{figure}
 \begin {center}
 \begin{tikzpicture}[scale=2]
  \node[cylinder,draw=black,thick,aspect=0.7,
        minimum height=1.7cm,minimum width=2.5cm,
        shape border rotate=-80,
        cylinder uses custom fill,
        ]  at (0,0)
   (A) {$SO_2\times{SO}_3/SO_2$};  
    \node[ellipse,draw=black,thick,aspect=0.7,
        minimum height=1.4cm,minimum width=2.5cm,
        shape border rotate=60,
        cylinder uses custom fill,
        ]  at (0,-1.9)
   (B) {$SO_3/SO_2$}; 
    \node [draw]  at (3,-1) {Trivialization : $p_1 + p_0\sigma_x = -imv\sigma_y$};
   \draw [draw=black,very thick] (-0.2,-0.6) -- (-0.2,-1.7);
\end{tikzpicture}
\end {center}
\caption {Fiber bundle formed in the neighborhood of (0,0,R) with the homogeneous space $SO_3/SO_2$ as the base and $SO_2$ as fibers.}
\end {figure}

The coin ensemble based on $SO_3$ constructed earlier can be described in the language of associated bundles as well. Starting with the principal G-bundle $\eta = (SO_3, \pi, SO_3/SO_2)$ and a vector space V on which $SO_2$ has a representation as the left G-space, fiber, let us construct the associated bundle $\eta[SO_2] = (SO_3\times_{V}V, \pi_{V}, SO_3/SO_2)$. The left action of $SO_2$ is defined by $g(p,v) = (ph, h^{-1}v), h\in{SO_2}$. This associated bundle has a cross-section, gauranteed by the theorem on page 149, \cite {Isham1989} S with respect to a mapping $\phi: P(\eta)\rightarrow{F}$ satisfying $\phi(pg) = g^{-1}\phi(p),\forall{p}\in{P(\eta)}, g\in{G} $  defined as
\begin {equation}
S_\phi(x) = [p, \phi(p)], \text{  where, p is any point on the fiber } \pi^{-1}(x).
\end {equation}
The cross-section can be used to build a Hilbert space containing rays of the wave functions $\Psi$ defined in equation \eqref {dsDirac}.

\pgfplotsset{
    colormap={whitered}{
        color(0cm)=(white);
        color(1cm)=(orange!75!red)
    }
}

\section {Summary and conclusions}
We derived the kinematics of split-step quantum walk using induced representations of groups and  expressed them in terms of systems of imprimitivity. We developed a geometric picture of the time-discrete version of the evolution as a fiber bundle in the momentum space with a 2-torus containing an hyperboloid as the base manifold. The continuous time counterpart is a bundle with the real plane as the base manifold and both the systems have the group $SL(2,\mathbb{C})$ (spinors) as the fibers. We also constructed systems of imprimitivity for walks based on $SO_3$ coins and derived Dirac equation on a de Sitter space. The automorphisms on distance-transitive graphs was shown as SI. Quantum walk framework being an important simulation tool it is imperative to cast the kinematics in terms of systems of imprimitivity to leverage the tools of induced representations for building more complex, for example gauges of cohomology, evolutions. We plan to build 2-cocycles from SI for the quantum walks and using that to construct a semigroup whose dilation results in a Fock space where the non-interactng walker evolves freely.


\end{document}